\documentclass[twoside,leqno,twocolumn]{article}  
\usepackage{ltexpprt} 
\usepackage{paralist}
\usepackage{algorithm}
\usepackage{algorithmic}
\usepackage{amsmath}
\usepackage{multirow}

\usepackage{subfigure}
\usepackage{graphicx}
\usepackage{epsfig} 

\usepackage{amssymb}

\begin{document}

\title{\Large When and Where: Predicting Human Movements Based on Social Spatial-Temporal Events}
\author{Ning Yang\thanks{College of Computer Science, Sichuan University, China. yangning@scu.edu.cn} \\
\and 
Xiangnan Kong \thanks{Department of Computer Science, University of Illinois at Chicago, USA. xkong4@uic.edu}\\
\and
Fengjiao Wang \thanks{Department of Computer Science, University of Illinois at Chicago, USA. fwang27@uic.edu}\\
\and
Philip S. Yu \thanks{Department of Computer Science, University of Illinois at Chicago, USA. psyu@uic.edu} \\}
\date{}

\maketitle


\begin{abstract} \small\baselineskip=9pt 
Predicting both the time and the location of human movements is valuable but challenging for a variety of applications. To address this problem, we propose an approach considering both the periodicity and the sociality of human movements. We first define a new concept, Social Spatial-Temporal Event (SSTE), to represent social interactions among people. For the time prediction, we characterise the temporal dynamics of SSTEs with an ARMA (AutoRegressive Moving Average) model. To dynamically capture the SSTE kinetics, we propose a Kalman Filter based learning algorithm to learn and incrementally update the ARMA model as a new observation becomes available. For the location prediction, we propose a ranking model where the periodicity and the sociality of human movements are simultaneously taken into consideration for improving the prediction accuracy. Extensive experiments conducted on real data sets validate our proposed approach.
\end{abstract}

\section{Introduction}

Recently, predicting human movements has been attracting considerable interest because of its great value in a variety of applications such as location-based services, monitoring of epidemic\cite{r26,r1,r2}. Existing researches share a common idea to model human movements based on sequential patterns discovered from historical trajectories\cite{r5,r11,r12,r20,r21,r22}, and have two major shortcomings. Firstly, existing approaches mainly focus on predicting the next locations only, but neglect the occurrence time. In many real-world applications, it is usually essential to know both the next location where a given person will appear, and the time when he/she will appear at that location. For example, in homeland security, it is very important to estimate both when and where a suspect will appear. In epidemic monitoring, accurate prediction of both when and where people will be infected by a kind of deadly virus is vital for the prevention of disease outbreak. Secondly, existing prediction models based on sequential patterns are not easy to understand while assuming the activities of different people are independent from each other.

In fact, more and more research evidences show that most of human movements are strongly periodic, and can be understood in the context of social relationships \cite{r1, r2, r4, r13, r18, r19, r26}. For example, when people go out, they are more likely to party with friends, or meet with colleagues than going to a random place, and these social interactions often take place regularly. These research evidences and observations suggest that the prediction of human movements can be largely reduced to the prediction of social interactions. 

Motivated by the above applications and observations, in this paper we aim at the problem of predicting both the time and the location of a given person's next movement, by considering his/her social relationship information. For this purpose, there are two main issues we have to address: 
\begin {compactenum} [(1)]

\item Social relationships are generally represented as graphs, while the time and the location of human movements are usually modeled as time series or sequences. So, how to utilize social information, the graph data, to facilitate the prediction of time and location, the sequence data?  

\item Intuitively, it is hard to estimate the time and the location at the same time, so, which prediction should be made first, and how does one help the other? 

\end {compactenum}

For the first issue, we propose a new concept, \emph{Social Spatial-Temporal Event} (\emph{SSTE}), to model social interactions between a given person and his/her friends. Essentially, an SSTE of a given person is a subgraph where nodes represent the given person and a subset of his/her friends. We associate each subgraph with two spatial-temporal attributes, i.e., the time when and the location where a social interaction occurs. SSTE integrates the social information and the spatial-temporal information of human movements, so the regularity of SSTEs of a given person can reveal the dynamics of social interactions of that person, which makes it reasonable and feasible to formulate the prediction of human movements as the time and location predictions of the next SSTE of a given person. 

For the second issue, our idea is in light of the plenty of researches that indicate the human movement exhibits a strong periodicity \cite{r1, r2} and a strong sociality \cite{r13}. The periodicity implies that the location depends on its occurrence time. For example, the places where people meet on weekday are reasonably different than those on weekend. The sociality implies that the location of an SSTE of a given person is also influenced by the preference of his/her friends. That is to say the place where an event occurs will be influenced by not only the friends or the participants, but also the time of occurrence. Hence making time prediction first will be better able to exploit these  properties. We thus make time prediction first, then use the predicted time combined with the social information as the input of the location prediction. 

\textbf{Time Prediction} The goal of the time prediction of an SSTE is to estimate how long it will take before the next SSTE involving a given person occurs. We characterize the temporal dynamics of SSTEs of a given person with an ARMA (AutoRegressive Moving Average) model \cite{r10}. The time prediction based on the ARMA model only depends on the mean value and the autocorrelation function of the time interval series. This is extremely important for our goal because of the difficulty of estimating all of the joint distributions on massive data. 

As user behavior can change from time to time, an ARMA model whose parameters keep constant will not be able to fully capture the behavior dynamics. Here we propose a learning algorithm which employs Kalman Filter \cite{r14} to learn and incrementally update the parameters of the prediction model as a new observation becomes available. Our learning algorithm requires only the latest observation to update the ARMA model, which is in sharp contrast with the traditional methods like MLE (Maximum Likelihood Estimation) that recalculate the parameters based on the whole data every time when a model update is required.

\textbf{Location Prediction} 
By combining the periodicity and sociality of human movements, we propose a ranking model by which we rank all the candidate locations and output the locations with higher ranking scores as the prediction result.

In summary, the contributions of this paper are as follows:
\begin {compactenum} [(1)]

\item We reduce the prediction of human movements to the prediction of social interactions causing those movements. For this purpose, we propose a new concept, Social Spatial-Temporal Event, to characterize the periodicity and sociality of human movements. 

\item For the prediction of the occurrence time of an SSTE, we propose an ARMA based prediction model and a Kalman Filter based learning algorithm which is capable of incrementally updating the prediction model as a new observation becomes available. 

\item We propose a ranking model for the prediction of an SSTE location, where the periodicity and the sociality of human movements are simultaneously taken into consideration for improving the prediction accuracy. 

\item The performance of our proposed methods are validated via the extensive experiments conducted on real data sets.

\end {compactenum}

The rest of this paper is organized as follows. The preliminary concepts and terminologies are introduced in Section 2. The concept of SSTE and its detection are described in Section 3. The time prediction model and algorithm are presented in Section 4. The ranking model and the location predicting algorithm are worked out in Section 5. We analyze the experimental results in Section 6. At last, we briefly review the related work in Section 7 and conclude in Section 8.

\section{Preliminaries}
In this section we briefly introduce the basic concepts of ARMA model and Kalman Filter.

\subsection {AutoRegressive-Moving-Average Model}

AutoRegressive-Moving-Average (ARMA) model is a powerful tool to analyze a stationary stochastic process. Given a time series of historical data $x^t = (x_1, x_2$, $\cdots, x_t)$, where $x_i (1 \le i \le t)$ is the data vector at time $i$, the ARMA model can predict the future value $\hat{x}_{t+1}$. An ARMA model consists of two parts, an autoregressive (AR) part and a moving average (MA) part. The AR model of order $p$, AR($p$), is given by the following equation:
\begin{equation}
    \Phi(B)(x_{t}) = 0,
\end{equation}
where 
\begin{displaymath}
	\Phi(z)=1-\phi_{1}z-\cdots-\phi_{p}z^p,
\end{displaymath}
and $B$ is the backward shift operator, $B^jx_{t} =x_{t-j}$.
The MA model of order $q$, MA($q$), is given by the following equation:
\begin{equation}
    \Theta(B)\varepsilon_{t} = 0,
\end{equation}
where
\begin{displaymath}
	\Theta(z)=1-\theta_{1}z-\cdots-\theta_{q}z^q,
\end{displaymath}
and $\varepsilon_{t}$ is a white noise with zero mean and variance $\sigma^2$.
Then by combining the equation (2.1) and (2.2), the ARMA model of order ($p,q$), can written as:
\begin{equation}
\Phi(B)(x_{t}) = \Theta(B)\varepsilon_{t}.
\end{equation}

\subsection {Kalman Filter}
Kalman filter is a set of mathematical equations that provides an efficient computational (recursive) means to estimate the state of a linear dynamic process from a series of noisy measurements, in a way that minimizes the mean of the squared error.
Kalman Filter assumes the true state at time $k$ evolves from the previous state at $(k-1)$, i.e.,
\begin{equation}
    x_k = Ax_{k-1}+Bu_{k-1}+w_{k-1},
\end{equation}
where
$u_{k-1}$ and $w_{k-1}$ are the control input and the process noise respectively, $A$ is the state transition matrix and $B$ is the control input matrix. $A$ might change over time, but here we assume it is constant. Since the control input is optional, so for simplicity, we assume $u_k=0$ for any $k$. The process noise $w_k$ is assumed
to follow a normal distribution with zero mean and covariance $Q_k$,
\begin{equation}
    p(w_k) \sim N(0, Q_k).
\end{equation}

At time $k$, an observation (or measurement) $z_k$ of the true state $x_k$ is
\begin{equation}
    z_k = Hx_k+v_k,
\end{equation}
where $H$ is the measurement matrix and $v_k$ is the observation noise which is also assumed to follow a normal distribution with zero mean and covariance $R_k$,
\begin{equation}
    p(v_k) \sim N(0, R_k).
\end{equation}

Our objective is to predict $\hat{x}_{k|k}$, the posteriori state estimate at time $k$ given observations up to $k$. This computation can be divided into two alternate phases, "Predicting" and "Updating". In the predicting phase, the priori state estimate $\hat{x}_{k|k-1}$ is produced from the posteriori state estimate at previous time step $\hat{x}_{k-1|k-1}$ \cite{r14}, i.e., 
\begin{equation}
    \hat{x}_{k|k-1} = A\hat{x}_{k-1|k-1}+Bu_{k-1},
\end{equation}
and
\begin{equation}
    P_{k|k-1} = AP_{k-1|k-1}A^T+Q_{k},
\end{equation}
where $P_{k|k-1}$ is the covariance matrix of the priori estimate.

In the updating phase, the current posteriori state estimate is generated by combining the current priori estimate with the current observation information. Kalman Filter has a set of formulas well-established for the update phase, respectively shown as follows \cite{r14}.

Innovation (or measurement residual):
\begin{equation}
    y_k = z_k-H\hat{x}_{k|k-1}.
\end{equation}
Innovation covariance:
\begin{equation}
    S_k = HP_{k_k-1}H^T+R_k.
\end{equation}
Optimal Kalman gain:
\begin{equation}
    K_k = P_{k|k-1}H^TS^{-1}_k.
\end{equation}
Posteriori state estimate:
\begin{equation}
    \hat{x}_{k|k} =\hat{x}_{k|k-1}+K_ky_k. 
\end{equation}
Posteriori estimate covariance:
\begin{equation}
    P_{k|k} = (I-K_kH)P_{k|k-1}.
\end{equation}

\section {Social Spatial-Temporal Event}

What LBSNs provide us contain a friendship graph and whole sequences of recorded checkins. The friendship graph is denoted by $G(V,E)$ where $G.V$ is the collection of persons and $G.E$ is the set of edges. There exists an edge between two persons if they are friends. The checkin sequence is denoted by $R=(r_1,r_2,\dots,r_n)$ where $r_i$ is the $i$th checkin in order of time. The concept of checkin and Social Spatial-Temporal Event are defined as follows:

\begin{Definition}
{\rm \textbf{Checkin} a checkin record $r$ is a 3-tuple $(v,\tau,\sigma)$, where $r.v \in G.V$ is the person who checks in, $r.\tau$ is the checkin time, $r.\sigma$ is the coordinates of the checkin location.}
\end{Definition}
 
\begin{Definition}
{\rm \textbf{Social Spatial-Temporal Event (SSTE)} an SSTE, denoted by $a$, is a 3-tuple $(G,\tau,\sigma)$, where $a.G$ is a subgraph of the friendship graph $G$, $a.\tau$ is the occurrence time of $a$, and $a.\sigma$ is the coordinates of the location where $a$ ocurrs.}
\end{Definition} 

There exist a lot of time-evolving clustering algorithms that can serve the detection of SSTEs from a checkin sequence. In this paper, we implement the SSTE detection by the density based temporal clustering algorithm proposed by Yixin Chen et al. \cite{r9}, because its time-evolving weight schema is suitable for the requirement to adaptively determine the temporal boundary between two successive SSTEs.

\section{Time Prediction}
Now we turn our attention to the time prediction of SSTE, which can be formalized as follow:

\textbf{Time Prediction} With the known SSTE sequence involving person $u$, denoted as 
\begin{displaymath}
	A^n_u=(a_{u,1},a_{u,2}, \dots, a_{u,n}), \forall a_{u,i}, 1 \le i \le n, u \in a_{u,i}.V
\end{displaymath} 
where $a_{u,i}$ is the $i$th SSTE of $u$, we want to predict the occurrence time of the next SSTE of $u$, i.e., $a_{u,n+1}.\tau$.

The time interval series $x^t_u$ can be derived from $A^n_u$, i.e. $x^t_u=(x_{u,1},x_{u,2}, \dots, x_{u,t})$, where $t = n-1$, $x_{u,i} = a_{u,i+1}.\tau-a_{u,i}.\tau, 1 \le i \le t$. Now the original problem turns out to be the prediction of $x_{u,t+1}$ with known $x^t_u$.

\subsection{Model Selection}

For a given person $u$, we can model $x^t_u$ by an ARMA($p,q$) model given by equation (2.3):
\begin{equation}
\Phi(B)(x_{u,t}) = \Theta(B)\varepsilon_{u,t},
\end{equation}
where 
\begin{displaymath}
	\Phi(z)=1-\phi_{u,1}z-\cdots-\phi_{u,p}z^p,
\end{displaymath}
\begin{displaymath}
	\Theta(z)=1-\theta_{u,1}z-\cdots-\theta_{u,q}z^q,
\end{displaymath}
and $B$ is the backward shift operator, $B^jx_{u,t} =x_{u,t-j}$, and $\varepsilon_{u,t}$ is a white noise with zero mean and variance $\sigma^2_u$.

The task of model selection is to determine the orders $p$ and $q$, which can be fulfilled by the following canonical procedure defined in time series analysis\cite{r10}:
\begin {enumerate} [(1)]

\item If the ACF (Auto-Covariance Function) of the $x^t_u$ cuts off after lag $q$ and the PACF (Partial ACF) tails off, $x^t_u$ is modeled as ARMA($0, q$);

\item If the PACF cuts off after lag $p$ and the ACF tails off, $x^t_u$ is modeled as ARMA($p, 0$);

\item If both the ACF and the PACF tail off, $x^t_u$ is modeled as ARMA($p,q$) where $p$ and $q$ are determined by the AIC (Akaike's Information Criterion) \cite{r10}.

\end {enumerate}

Note that $x^t_u$ should be substituted with its first or even higher order difference series if $x^t_u$ is not stationary, since the above procedure can apply only to a stationary process. 

\subsection{Learning the Model}

Once the model's orders are determined, the parameters, $\phi_{u,1}, \cdots, \phi_{u,p}$, $\theta_{u,1}, \cdots, \theta_{u,q}$ have to be learned. To make Kalman Filter applicable to our situation, the challenge here is to establish our own \emph{measurement equation} and \emph{state equation} \cite{r14}. 

\textbf{Measurement Equation} By equation (2.1), we have
\begin{equation}
x_{u,t} = H_{u,t}^T\Phi_{u,t}+\Theta^T_u\boldsymbol{\varepsilon}_{u,t},
\end{equation}
where 
\begin{displaymath}H_{u,t}=[x_{u,t-1},x_{u,t-2},\cdots,x_{u,t-p}]^T,\end{displaymath}
\begin{displaymath}\Phi_{u,t}=[\phi_{u,1},\cdots, \phi_{u,p}]^T,\end{displaymath}
\begin{displaymath}\Theta_u=[1, -\theta_{u,1}, \cdots, -\theta_{u,q}]^T,\end{displaymath}
\begin{displaymath}\boldsymbol{\varepsilon}_{u,t} = [\varepsilon_{u,t},\varepsilon_{u,t-1}, \cdots, \varepsilon_{u,t-q}]^T.\end{displaymath}
In light of the theory of Bayesian statistics, our idea here is to regard the parameter vector $\Phi_{u,t}$ as the random vector representing the state of the underlying system, and $x_{u,t}$ as the observable variable (i.e., $\Phi_{u,t}$ plays the role of $x_k$ in equation (2.4) and $x_{u,t}$ plays the role of $z_k$ in equation (2.6)). We hence let equation (4.16) be the measurement equation of Kalman Filter, where time-varying matrix $H_{u,t}$ is the measurement matrix, and $\Theta^T_u\boldsymbol{\varepsilon}_{u,t}$ is the measurement noise with zero mean and the variance 
\begin{displaymath}r_u = (1+\sum_{i=1}^q{\theta_{u,i}^2})\sigma^2_u.\end{displaymath}

\textbf{State Equation} The state equation is a bit less obvious. Considering equation (4.15) is an ARMA model, we deliberately establish the state equation as
\begin{equation}
\Phi_{u,t}= \Phi_{u,t-1} + W_{u,t},
\end{equation}
where the $p \times 1$ matrix $W_{u,t}$ is the process noise which is assumed to be independent, white, and with zero mean and the covariance $Q_u$, a $p \times p$ matrix.  

Let $\hat{\Phi}_{u,t}$ be the estimate of ${\Phi}_{u,t}$, $e_{u,t}$ be the estimate error, 
\begin{displaymath}e_{u,t}=\Phi_{u,t}-\hat{\Phi}_{u,t}.\end{displaymath}
Let $P_{u,t}$ be the covariance matrix of the estimate error, i.e., 
\begin{displaymath}P_{u,t} = E[e_{u,t}e_{u,t}^T].\end{displaymath}
By combining the time update procedure and the measurement update procedure in standard Kalman Filter algorithm \cite{r14}, we have the following iterative equations for the parameter update when a new observation $x_{u,t}$ becomes available:
\begin{equation}
K_{u,t+1} = (P_{u,t}+Q_u)H_{u,t} [H^T_{u,t}(P_{u,t}+Q_u)H_{u,t}+r_u]^{-1}
\end{equation}
\begin{equation}
P_{u,t+1} = (I-K_{u,t+1}H_{u,t}^T)(P_{u,t}+Q_u)
\end{equation}
\begin{equation}
\hat{\Phi}_{u,t+1} = \hat{\Phi}_{u,t}+K_{u,t+1}(x_{u,t}-H_{u,t}^T\hat{\Phi}_{u,t}),
\end{equation}
where $I$ is the unit matrix and $K_{u,t+1}$ is the Kalman gain, a $p \times 1$ matrix. 

\textbf{Initialization} We use traditional MLE once to initialize $\hat{\Phi}_{u,0}$, $\Theta_u$, and $\sigma^2_u$. As for any $P_{u,0} \ne 0$ and for any $Q_u$ the estimate given by equation (4.20) would eventually converge \cite{r14}, we simply presume $P_{u,0} = 1_{p \times p}$ and $Q_u = 0_{p \times p}$.

As the result of this subsection, the parameter learning algorithm is given in Algorithm 1.

\begin{algorithm}[!ht]
  \caption{ \emph{LearnByKF}$(\hat{\Phi}_{u,t}, x_{u,t}, H_{u,t}, P_{u,t})$ }
  \label{alg:LearnByKF}
  \begin{algorithmic}[1]
    \REQUIRE ~~ \\
       $\hat{\Phi}_{u,t}$: the estimator of the last time; \\
       $x_{u,t}$: the new observation; \\
       $H_{u,t}$: the historical observations; \\
       $P_{u,t}$: the covariance of estimate error at time $t$;
    \ENSURE ~~
       $\hat{\Phi}_{u,t+1}$: the new estimate of parameters;
    \STATE Update the Kalman gain $K_{u,t+1}$ by equation (4.18);
    \STATE Update the covariance $P_{u,t+1}$ by equation (4.19);
    \STATE Generate the new estimate $\hat{\Phi}_{u,t+1}$ by equation (4.20);
  \end{algorithmic}
\end{algorithm}

\subsection{Making Time Prediction}

In this paper, we define the best estimate of $x_{u,t+1}$, $\hat{x}_{u,t+1}$, to be the predictor with minimum mean square error, i.e., given the history $X^t_u$,
\begin{displaymath}
	\hat{x}_{u,t+1}=E[X_{u,t+1}|X^t_u].
\end{displaymath}
where $X^t_u$ represents the random sequence corresponding to $x^t_u$, i.e., $X^t_u=(X_{u,1},X_{u,2}, \dots, X_{u,t})$. Concretely, we can compute $\hat{x}_{u,t+1}$ in terms of the following theorem.

\begin{theorem}
The best estimate of $x_{u,t+1}$ is 
\begin{equation}
\hat{x}_{u,t+1} = \hat{\Phi}_{u,t+1}^TH_{u,t+1}-\sum_{i=1}^{q}{\theta_{u,i}\varepsilon_{u,t-i+1}}.
\end{equation}
\end{theorem}
\begin{proof}
By equation (4.16),
\begin {displaymath}
X_{u,t+1} = \sum_{i=1}^{p}{\hat{\phi}_{u,i}X_{u,t-i+1}}+\varepsilon_{u,t+1}-\sum_{i=1}^{q}{\theta_{u,i}\varepsilon_{u,t-i+1}}.
\end {displaymath}
So,
\begin{displaymath}
\begin{split}
\hat{x}_{u,t+1} =& \sum_{i=1}^{p}{\hat{\phi}_{u,i}E[X_{u,t-i+1}|X^t_u]}+E[\varepsilon_{u,t+1}|X^t_u] \\
&-\sum_{i=1}^{q}{\theta_{u,i}E[\varepsilon_{u,t-i+1}|X^t_u]}.
\end{split}
\end{displaymath}
Since $\varepsilon_{u,t}$ and $\varepsilon_{u,t-i+1} (1 \le i \le q)$ are independent of $X^t_u$, we have 
\begin{displaymath}
	E[X_{u,t-i+1}|X^t_u] = x_{u,t-i+1},
\end{displaymath}
\begin{displaymath}
	E[\varepsilon_{u,t+1}|X^t_u]=0, 
\end{displaymath}
\begin{displaymath}
	E[\varepsilon_{u,t-i+1}|X^t_u] = \varepsilon_{t-i+1}.
\end{displaymath}
Hence substituting them into the above equation gives equation (4.21).
\end{proof}

Since $\hat{x}_{u,t+1}$ is just the estimate of the time-gap between $a_{u,t}$ and $a_{u,t+1}$, the occurrence time of  $a_{u,t+1}$ will be further estimated by 
\begin{equation}
a_{u,t+1}.\hat{\tau}=a_{u,t}.\tau+\hat{x}_{u,t+1}.
\end{equation}
Algorithm 2 gives the complete procedure of time prediction.

\begin{algorithm}[!ht]
  \caption{ \emph{AlgKF}$(\hat{\Phi}_{u,t}, x^{t}_{u}, H_{u,t}, P_{ut})$ }
  \label{alg:AlgKF}
  \begin{algorithmic}[1]
    \REQUIRE ~~ \\
       $\hat{\Phi}_{u,t}$: the last estimate of the parameters; \\
       $x^{t}_u$: the historical observations; \\
       $H_{u,t}$: the measurement matrix; \\
       $P_{u,t}$: the covariance of estimate error at time $t$;
    \ENSURE ~~
       $a_{u,t+1}.\hat{\tau}$: the estimate of the occurrence time of $a_{u,t+1}$;
    \STATE $\hat{\Phi}_{u,t+1}=LearnByKF(\hat{\Phi}_{u,t}, x_{u,t}, H_{u,t}, P_{u,t})$;
    \FOR {$i=1$ to $q$}
     	\STATE $\varepsilon_{u,t-i+1}=x_{u,t-i+1}-\hat{x}_{u,t-i+1}$;
    \ENDFOR
    \STATE Compute $\hat{x}_{u,t+1}$ in terms of equation (4.21);
    \STATE Compute $a_{u,t+1}.\hat{\tau}$ in terms of equation (4.22);
  \end{algorithmic}
\end{algorithm}

Note that if $x^t_u$ is substituted with its difference series in equation (4.15), we have to further correct $\hat{x}_{u,t+1}$ with the differences just before line 6.

\begin{figure*}[!ht]
\centering
    \begin{minipage}[b]{0.32\textwidth}
        \centering
         \subfigure[CHI]{\epsfig{file=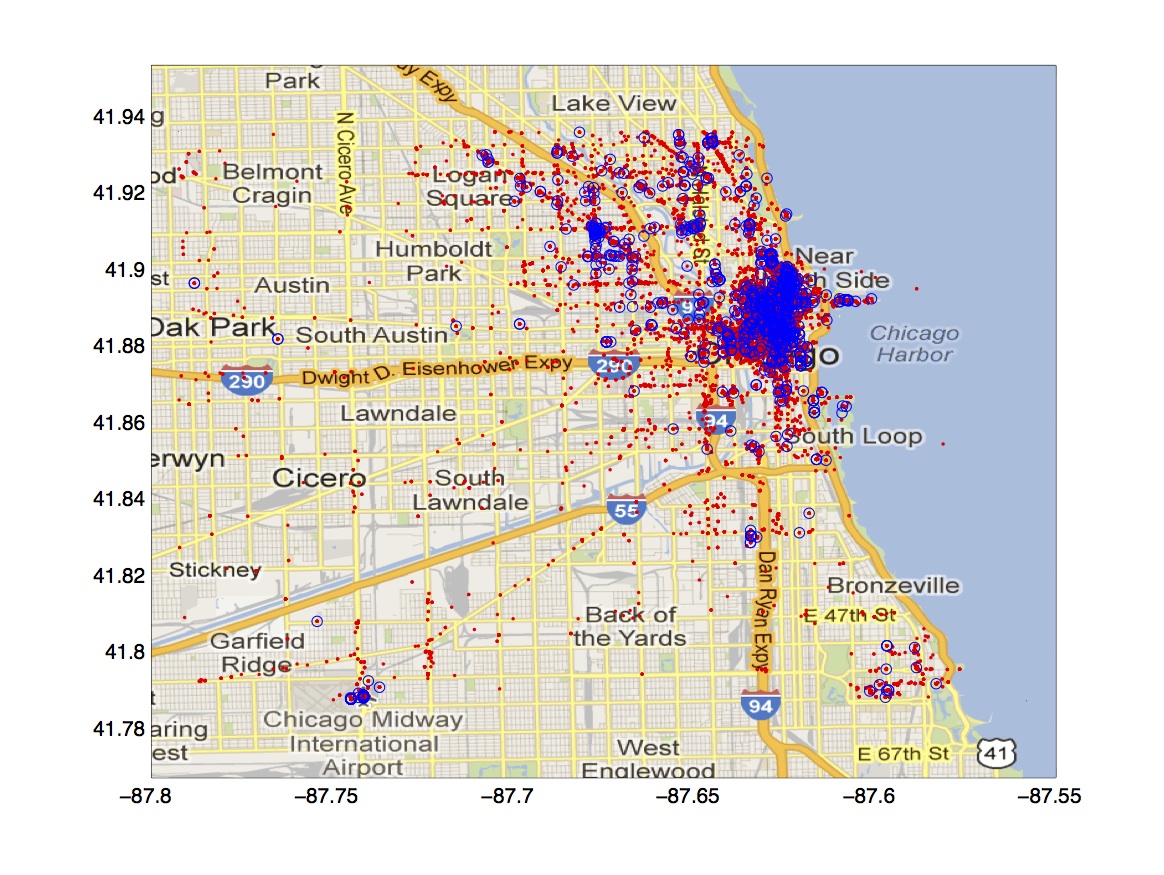, height=1.6in}}
    \end{minipage}
    \begin{minipage}[b]{0.32\textwidth}
         \centering
         \subfigure[NY]{\epsfig{file=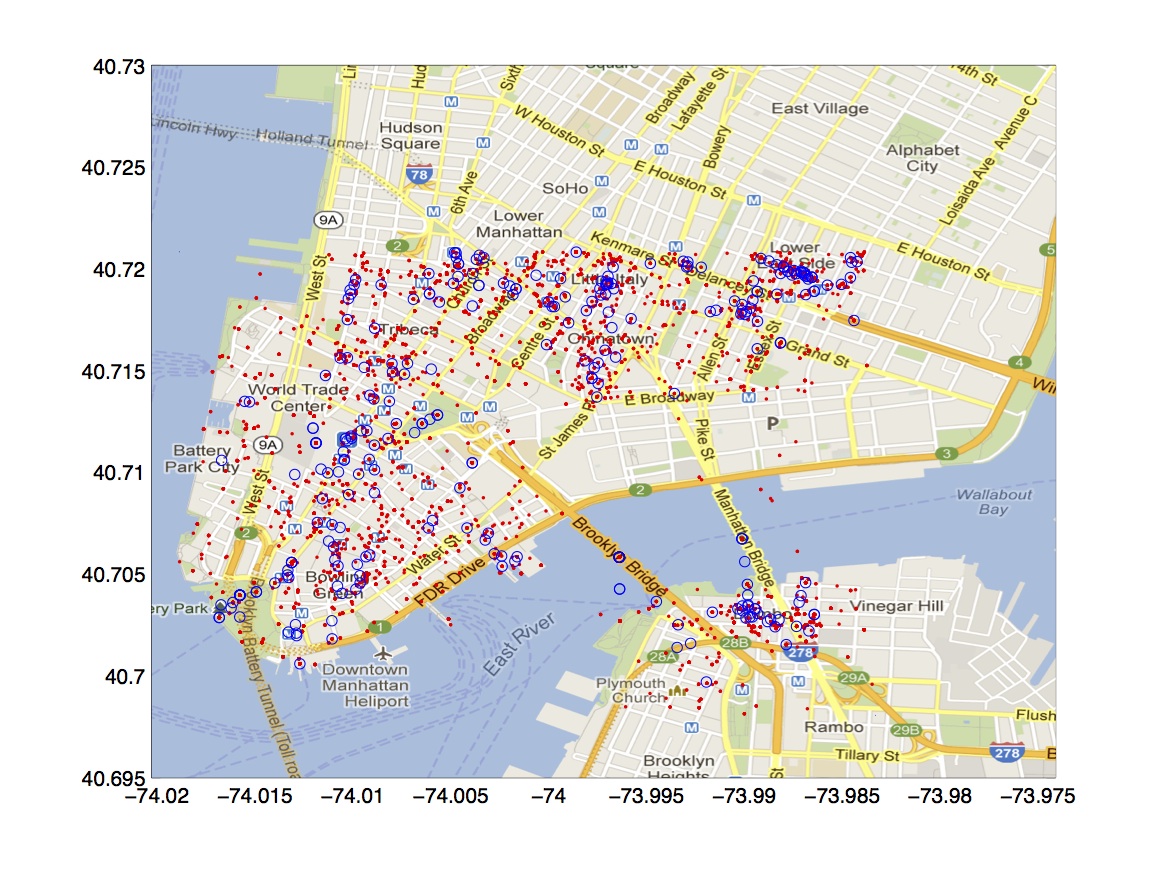, height=1.6in}}
    \end{minipage}
    \begin{minipage}[b]{0.32\textwidth}
         \centering
         \subfigure[LV]{\epsfig{file=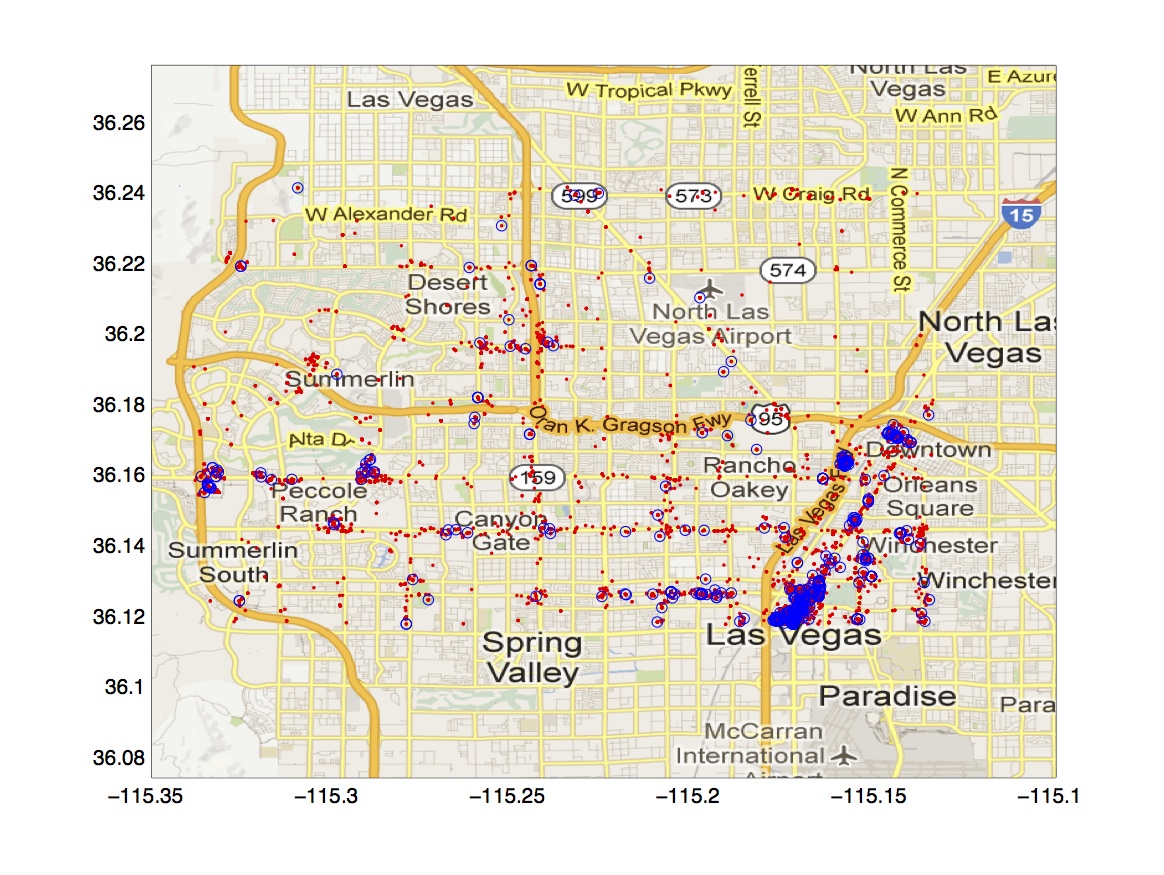, height=1.6in}}
    \end{minipage}\\
\caption{Geographic distribution of checkins and SSTEs.}
\end{figure*}

\section{Location Prediction}
The problem of the location prediction can be formalized as follow:

\textbf{Location Prediction} Given $A^n_u$, $u$'s friend set $F(u)$, and the predicted occurrence time $a_{u,n+1}.\hat{\tau}$, predict the location $a_{u,n+1}.\sigma$.

In order to predict the location at the future time point $a_{u,n+1}.\hat{\tau}$ predicted by the Kalman Filter based algorithm described in Section 3, we explore the concepts of the periodic behavior and the influence of friends. 

First, we consider the periodic behavior. We use one week as the period and one hour as the granularity and convert $a_{u,n+1}.\hat{\tau}$ to the hour of a weekday. For example, if the predicted time is noon next Friday, we will use the noon time locations of previous Fridays to help make the prediction. Let $L_u$ and $T_u$ be the random variables respectively representing the location and the time of an SSTE involving $u$. We define the temporal score of a candidate location $\lambda$ as:
\begin{equation}
\begin{split}
g_{temporal}(\lambda) &= P(L_u=\lambda|T_u=\hat{\tau}) \\
&=\frac{P(T_u=\hat{\tau}|L_u=\lambda)P(L_u=\lambda)}{P(T_u=\hat{\tau})},
\end{split}
\end{equation}
where $\hat{\tau} = a_{u,n+1}.\hat{\tau}$, and the probabilities can be respectively estimated by the following equations:
\begin{displaymath}
\begin{split}
P(T_u=\hat{\tau})=\frac{|\{a: a\in A^n \land a.\tau=\hat{\tau} \}|}{n},
\end{split}
\end{displaymath}
\begin{displaymath}
P(L_u=\lambda)=\frac{|\{a: a\in A^n \land a.\sigma=\lambda \}|}{n},
\end{displaymath}
\begin{displaymath}
P(T_u=\hat{\tau}|L_u=\lambda) = \frac{|\{a: a\in A^n \land a.\tau = \hat{\tau} \land a.\sigma=\lambda \}|}{|\{a: a\in A^n \land a.\sigma=\lambda \}|}.
\end{displaymath}

Next, we explore the influence of friends, which means the location of $a_{u,n+1}$ is also dependent on the locations where any of $u$'s friends previously checked in before $a_{u,n+1}.\hat{\tau}$. In addition, it is reasonable that more recent checkins are more influential. Hence we define the social score of a candidate location $\lambda$ as
\begin{equation}
g_{social}(\lambda) = \frac{W_F}{W},
\end{equation}
where $W_F$ is the total weight of $u$'s friends who previously checked in $\lambda$ before $\hat{\tau}$,
\begin{displaymath}
	W_F= \sum_{r \in R \land r.v \in F(u) \land r.\sigma=\lambda }r.w(\hat{\tau}),
\end{displaymath}
 and $W$ is the total weight of all of $u$'s friends,
\begin{displaymath}
	W= \sum_{r \in R \land r.v \in F(u)}r.w(\hat{\tau}).
\end{displaymath}

By combining equations (4.9) and (4.10),  the overall ranking score of a candidate location is defined as
\begin{equation}
g(\lambda) = \xi g_{temporal}(\lambda) + (1-\xi)g_{social}(\lambda),
\end{equation}
where $\xi \in [0,1]$ is the weight reflecting how important the periodicity is. In practice, we choose all the locations where $u$ and any one of $u$'s friends previously checked in before $a_{u,n+1}.\hat{\tau}$ as the candidate location collection.

At last, note that it is hard and pointless to predict the accurate coordinates of an SSTE location. So as a preprocessing step before making location prediction, we discretize the coordinates of locations to regions by Voronoi algorithm \cite{r25}, where the regions are defined by the locations of bus stops. Then we practically predict the region instead of the coordinates. 

\section{Experimental Evaluation}
The objectives of experiments are to verify the validity of the SSTE detection algorithm and respectively evaluate the prediction models for SSTE time and location. All the experiments are conducted on a PC with Intel Core I7 CPU 2.0G HZ, 8 GB main memory. All the algorithms are implemented in MATLAB 2011b. 

\begin{figure*}[!ht]
\centering
    \begin{minipage}[b]{0.32\textwidth}
        \centering
         \subfigure[CHI]{\epsfig{file=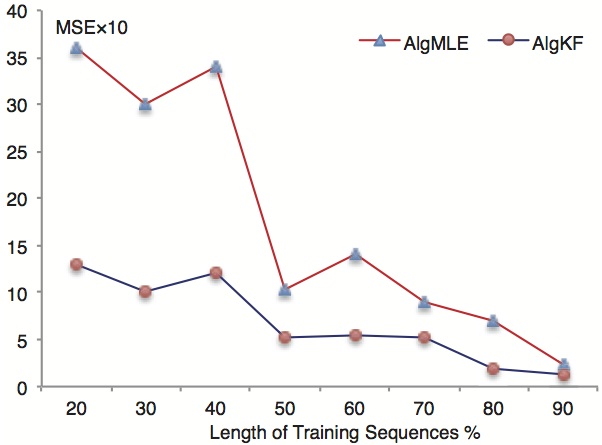, height=1.6in}}
    \end{minipage}
    \begin{minipage}[b]{0.32\textwidth}
         \centering
         \subfigure[NY]{\epsfig{file=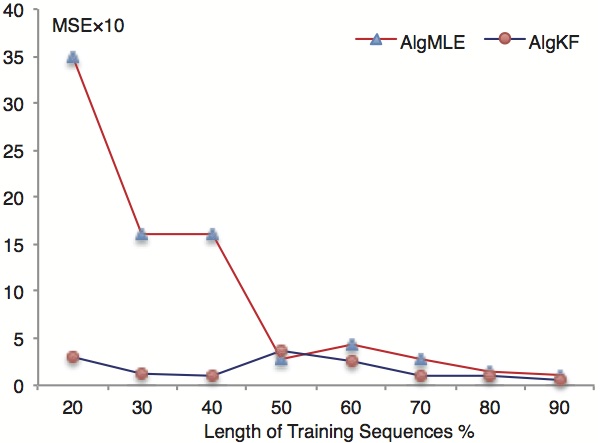, height=1.6in}}
    \end{minipage}
    \begin{minipage}[b]{0.32\textwidth}
         \centering
         \subfigure[LV]{\epsfig{file=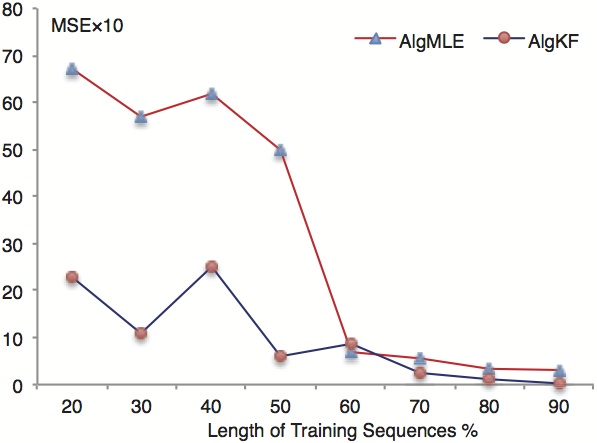, height=1.6in}}
    \end{minipage}\\
\caption{Error of time prediction.}
\end{figure*}

\textbf{Data Sets} We use three real checkin data sets and the corresponding friendship graphs offered by \cite{r16}, whose checkin areas are around Chicago (CHI), New York (NY) and Las Vegas (LV) respectively. CHI, NY and LV respectively contain 38,830, 2296, and 17,042 social checkins in one year. The data distributions are shown in Figure 1 (a), (b) and (c), where a red dot represents a checkin and a blue circle an SSTE.

\subsection{Test of Time Prediction}
We compare AlgKF with AlgMLE, an algorithm we develop for the test. AlgMLE also makes prediction in terms of  the equation (4.22) but the parameters keep constant all time once they have been learned by using MLE. AlgKF and AlgMLE run over CHI, NY and LV with different length proportions of training sequences ranged from 20\% to 90\%, respectively. We measure the accuracy by the average MSE (Mean Squared Error) on each training data set. The results are shown in Figure 2, from which we can have the following observations and analyses:
\begin {compactenum} [(1)]
\item MSE decreases with the increased length proportion of training sequences, whichever algorithm runs. This corresponds with our expectation that more historical data are used, more accurate predictions can be made.

\item The MSE incurred by AlgKF is generally lower than that incurred by AlgMLE at different length proportions of training sequences. This demonstrates that AlgKF outperforms the algorithms that keep the prediction model parameters constant because AlgKF is able to dynamically update the prediction model by Kalman Filter.   

\item The MSE difference between AlgKF and AlgMLE is more significant at less length proportions of training sequences. The underlying reason is that although the initial accuracy is low due to the small training data, AlgKF is able to dynamically learn from the errors in virtue of Kalman Filter, while AlgMLE is unable. So AlgKF can quickly and continuously improve the prediction accuracy as more and more observations come into sight, which results in that the final MSE would not be too large. 
\end {compactenum}

\begin{figure}
\centering
\epsfig{file=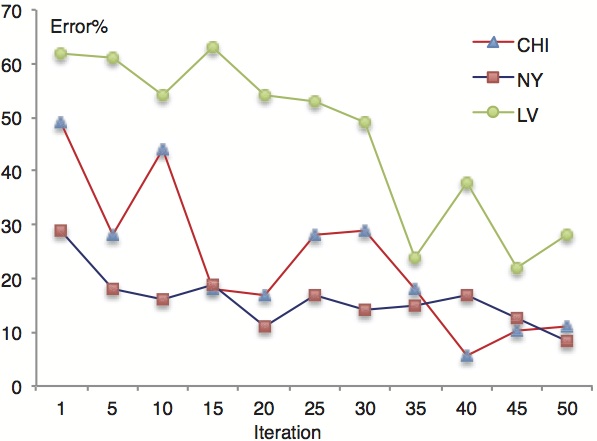, height=1.6in}
\caption{Converging speed of error.}
\end{figure}

Figure 3 further illustrates the error convergence speed on three particular SSTE sequences chosen from CHI, NY and LV respectively, where the proportion of training data is 20\%. As we can see, the error drops sharply after about 30 iterations compared with the error after first iteration.

\subsection{Test of Location Prediction}

We use \emph{Accuracy@TopN}, a metric proposed by \cite{r3}, to assess the performance of the location prediction model (equation (12)). In terms of \emph{Accuracy@TopN}, an SSTE location is predicted successfully if it is ranked in the Top-$N$ list. 

Figure 4 shows the average accuracy over the three data sets, at different Top-$N$ list sizes and at different values of the weight $\xi$, where the length proportion of training sequences is set to 80\%. At first, we can observe that the accuracy is increasing with increased $N$, no matter what value of $\xi$ is, which logically corresponds with our expectation. Figure 4 also shows different values of $\xi$ have different impacts on the accuracy. When $\xi=0$, i.e., the location of the next SSTE of a person $u$ is predicted completely by the preference of $u$'s friends, the accuracy is totally unacceptable, since even $Accuracy@Top50$ is only around 10\%. With increased $\xi$, i.e., more consideration is given to the periodicity, the accuracy increases, until $\xi=0.8$. When $\xi=1$, i.e., the location prediction is made completely in terms of the periodicity, the accuracy is lower than the accuracy at $\xi=0.8$. So, we can draw a conclusion that the location of an SSTE depends on both the periodicity and the sociality, but more on periodicity, which verifies the assumption we make in Section 1.   

\begin{figure}
\centering
\epsfig{file=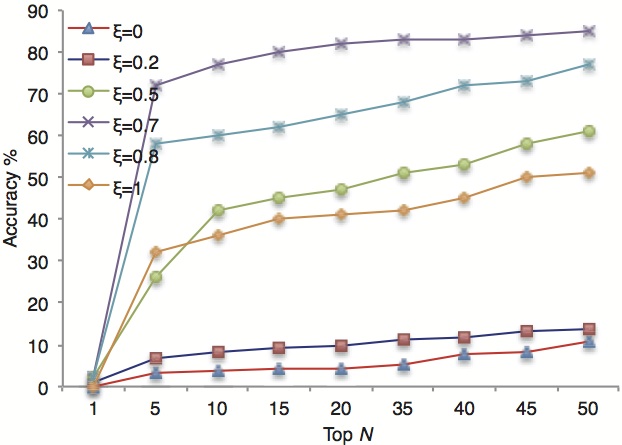, height=1.6in}
\caption{$Accuracy@TopN$ on 80\% training data.}
\end{figure}

At last, we investigate how the accuracy of location prediction changes with the length of training SSTE sequences. We show the $Accuracy@Top5$ and the $Accuracy@Top20$ at different lengths of training sequences in Figure 5(a) and (b) respectively. As we can observe from Figure 5, the accuracy generally increases with the increased length proportion of training sequences, which implies that longer the training sequence, more accurate the time prediction, and consequently more accurate the location prediction. 

\begin{figure}
\centering
\subfigure [$Accuracy@Top5$]{
\epsfig{file=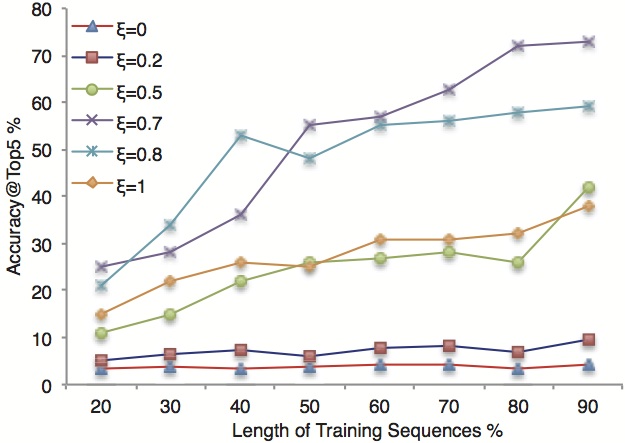, height=1.6in}
}
\hspace{0.5in}
\subfigure [$Accuracy@Top20$]{
\epsfig{file=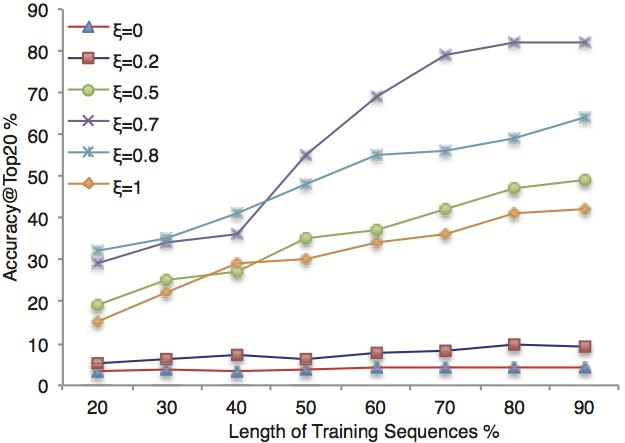, height=1.6in}
}
\caption{Accuracy at different lengths of training sequences.}
\end{figure}

\section{Related Work}
There are three domains related to our work, Spatial-Temporal Event Mining, Time Prediction, and Location Prediction.

\textbf{Spatial-Temporal Event Mining} Hady W. Lauw et al. \cite{r8} propose a spatial-temporal event model to discover social associations among individuals based on spatial-temporal co-occurrences. The concept of the spatial-temporal event proposed by \cite{r8} is different from ours since it allows an individual to participate more than one event at the same time. Xingjie Liu et al. \cite{r17} identify a new type of social network, event-based social network (EBSN) and employ an extended Fiedler method for the community detection. 

\textbf{Time Prediction} To our best knowledge, few of the previous studies focus on the prediction of the SSTE time. In \cite{r15}, Yizhou Sun et al. propose a generalized linear model based model to predict the first building time of the link between two entities who have never been linked before, rather than how long the time is before the existing entities will be linked again in the future, while the latter is exactly the objective for the time prediction of our work. Zan Huang el al. \cite{r7} proposes a time series model to evaluate the frequency of repeated links, which focuses not on the prediction of the time when a link would occur, but on the prediction of whether a link would occur.  

\textbf{Location Prediction} Location prediction has been received extensive attention in the field of trajectory analysis. Most of them are based on sequential pattern \cite{r4, r5, r11, r12}, where the time just play the role of a time-stamp. Eunjoon Cho et al. \cite{r1} and Dashun Wang et al. \cite{r2} find that movements of humans are periodic and correlated with their social network. Eunjoon Cho et al. further reveal that human movements can be explained by periodic behavior more than by social relationships, and propose a periodic and social mobility model for predicting mobility of individuals, which focuses on predicting the location of a single checkin, rather than the location of a social interaction with at least two participants. Meanwhile it is based on the time given by user, not the time predicted in advance.

Note that our work is different with the traditional event prediction. The traditional event prediction focuses on the problem of whether an event will occur at a given time point, while our work focuses on both the occurrence time and the occurrence location. 

\section{Conclusions}

In this paper, we reduce the problem of human movement prediction to the prediction of SSTEs. For the time prediction, we propose an ARMA based prediction model and a Kalman Filter based learning algorithm. Our proposed time prediction algorithm requires only the latest data and can hence incrementally update the prediction model as a new observation becomes available. For the location prediction, we propose a ranking model to predict the location of the SSTE at the predicted occurrence time, where the periodicity and the sociality of SSTEs are simultaneously taken into consideration for improving the prediction accuracy. The experimental results demonstrate that our proposed algorithms can predict the time and location of an SSTE in an acceptable accuracy.

\section*{Acknowledgments}
This work is supported by the National Science Foundation of China under Grant Nos. 61173099, 61103043 and the Doctoral Fund of Ministry of Education of China under Grant No. 20110181120062. This work is also supported in part by NSF through grants CNS-1115234, DBI-0960443, and OISE-1129076, and US Department of Army through grant W911NF-12-1-0066.

\bibliographystyle{abbrv}
\bibliography{SSTE_SDM} 

\end{document}